\documentclass{article}
\usepackage{amssymb, amsmath, amsthm}
\usepackage{cite}

\newtheorem{theorem}{Theorem}[section]
\newtheorem{lemma}[theorem]{Lemma}
\newtheorem{proposition}[theorem]{Proposition}
\theoremstyle{remark}
\newtheorem{remark}[theorem]{Remark}
\numberwithin{equation}{section}

\allowdisplaybreaks[3]

\makeatletter
\renewcommand{\@fnsymbol}[1]{\ensuremath{%
   \ifcase#1\or * \or 1\or 2\or 3\or
   \mathsection\or \mathparagraph\or \|\or \star\or
   \star\star\or {\star\star}\star \else\@ctrerr\fi}}
\makeatother

\title{Glauber Dynamics in Continuum: A Constructive Approach to Evolution of States
\thanks{This work was financially supported by the DFG through SFB 701:
``Spektrale Strukturen und Topologische Methoden in der Mathematik"
and through the research projects 436 POL 125/0-1, 436 UKR 113/97,
436 UKR 113/98, which is cordially acknowledged by the authors.}}
\author{Dmitri Finkelshtein%
\thanks{Institute of Mathematics, National Academy of Science of
Ukraine, 01601 Kiev-4, Ukraine; e-mail: \texttt{fdl@imath.kiev.ua}}\and Yuri Kondratiev%
\thanks{Fakult\"at f\"ur Mathematik, Universit\"at Bielefeld, Postfach
110 131, 33501 Bielefeld, Germany;
e-mail: \texttt{kondrat@math.uni-bielefeld.de}}\and Yuri Kozitsky%
\thanks{Instytut Matematyki, Uniwersytet Marii Curie-Sk{\l}odwskiej,
30-031 Lublin, Poland; e-mail:
\texttt{jkozi@hektor.umcs.lublin.pl}}}
\begin{document}
\maketitle

\begin{abstract}
The evolutions of states is described corresponding to the Glauber
dynamics of an infinite system of interacting particles in
continuum. The description is conducted on both micro- and
mesoscopic levels. The microscopic description is based on solving
linear equations for correlation functions by means of an
Ovsjannikov-type technique, which yields the evolution in a scale of
Banach spaces. The mesoscopic description is performed by means of
the Vlasov scaling, which yields a linear infinite chain of
equations obtained from those for the correlation function. Its main
peculiarity is that, for the initial correlation function of the
inhomogeneous Poisson measure, the solution is the correlation
function of such a measure with density which solves a nonlinear
differential equation of convolution type.
\end{abstract}

\noindent {\it Key words:} Glauber dynamics, birth-and-death
process, point process, configuration space, correlation function,
scale of Banach spaces, Ovsjannikov method, Vlasov scaling, scaling
limit, nonlocal equation

\vskip.2cm

\noindent {\it MSC (2010):} 82C22; 60K35; 35Q83.

\section{Introduction}

In the statistical theory of large systems~\cite{Dob}, the system
states are described as probability measures on the corresponding
phase space rather than pointwise, which is typical for the standard
theory of dynamical systems. For such large systems, in order to
obtain the description independent of the system size one employs
the models where the system is infinite and distributed over a
noncompact manifold with positive density. A particular case
constitute models of interacting point particles distributed over
$\mathbb{R}^d$, which are widely used in mathematical physics,
ecology, sociology, etc,
see~\cite{Nancy,Dima,Dima0,Dima1,Dima2,Oles1,Oles,OlesLena}. Here
the states are probability measures on the space of the particle
configurations
\begin{equation} \label{C1}
 \Gamma\equiv\Gamma(\mathbb R^d):=\{\gamma\subset\mathbb R^d :
 |\gamma\cap K|<\infty\text{ for any compact $K\subset\mathbb R^d$
 }\},
\end{equation}
where $|A|$ denotes the cardinality of $A$. The system is characterized by a collection of appropriate functions
$F:\Gamma \rightarrow \mathbb{R}$, called
{\it observables}. For a state
$\mu$, the quantity
\begin{equation*}
 \langle \! \langle F, \mu \rangle \! \rangle = \int_{\Gamma} F (\gamma) \mu(d \gamma)
\end{equation*}
is called the mean value of observable $F$ in state $\mu$. Then the system evolution
is described as the evolution of observables
obtained from the Kolmogorov equation
\begin{equation}
 \label{R2}
\frac{d}{dt} F_t = L F_t , \qquad F_t|_{t=0} = F_0, \qquad t>0,
\end{equation}
where the `generator' $L$ is specified within the choice of the model.
The evolution of states
is obtained from the
Fokker--Planck equation
\begin{equation}
 \label{R1}
\frac{d}{dt} \mu_t = L^* \mu_t, \qquad \mu_t|_{t=0} = \mu_0,
\end{equation}
related to (\ref{R2}) by the duality
\begin{equation}
 \label{R}
\langle \! \langle F_0, \mu_t \rangle \! \rangle = \langle \!
\langle F_t, \mu_0 \rangle \! \rangle.
\end{equation}
Note that $L$ ought to be Markov in order that the solutions of
(\ref{R1}) be probability measures. One of the possibilities here is
to describe the evolution pathwise---by constructing a stochastic
Markov process $X^{\mu_0}_t$, corresponding to the `generator' $L$
and to the initial state $\mu_0$. Then the state $\mu_t$ is just the
distribution law of $X^{\mu_0}_t$. However, for a number of
important models this way encounters serious problems and hence is
rather unrealistic. Moreover, the mere existence of the process
tells not too much about the properties of the system evolution.
Thus, the main idea which we realize in this work is to describe the
system evolution as the Markov evolution of states $\mu_0 \mapsto
\mu_t$, not necessarily based on the pathwise description, and
accompanied with a more detailed study of its properties. In a
sense, our approach is suggested by
 classical works on the Hamiltonian dynamics where
 the system evolution is described as the evolution
of the corresponding correlation functions obtained by solving the
equation
\begin{equation}
 \label{R4}
\frac{d}{dt} k_t = L^\Delta k_t, \qquad k_t|_{t=0} = k_0.
\end{equation}
In the Hamiltonian dynamics, the analog of (\ref{R4}) is the BBGKY
hierarchy. As mentioned in~\cite{Dob}, kinetic equations of the
Hamiltonian dynamics
 allow one to describe the evolution approximately but in more detail and in
simpler terms. Such kinetic equations can be obtained from equations like (\ref{R4}),
 provided
all necessary information about their solutions is available, see
the corresponding discussion in~\cite[paragraph 6]{Dob}.

In the present article, we describe the Markov evolution of states
in terms of the correlation functions on both microscopic, obtained
from (\ref{R4}), and mesoscopic levels. The latter will be done by
means of a nonlinear (kinetic) equation obtained from (\ref{R4}) in
the Vlasov scaling limit. As in~\cite{Dima0,OlesLena}, our object is
the Glauber dynamics described by the `generator' (\ref{R2}) having
the form
\begin{align} \label{R20}
(L F)(\gamma) = & \sum_{x\in \gamma} \left[F(\gamma\setminus x) -
F(\gamma) \right]\\ & +  \varkappa \int_{\mathbb{R}^d}
 \exp\left(- \sum_{y\in \gamma} \phi(x-y) \right) \left[F(\gamma\cup x) - F(\gamma) \right] dx .\nonumber
\end{align}
Here the first term describes the particle death with constant rate,
whereas the second one is the birth term with activity $\varkappa>0$
and an interaction potential $\phi\geq 0$, which is supposed to obey
a natural integrability condition only. In contrast
to~\cite{Dima0,OlesLena}, where $\varkappa$ and $\phi$ were subject
to a certain constraint, here we obtain the evolution $k_0\mapsto
k_t$ for all $\varkappa$ and $\phi$, which, however, is restricted
to a limited time interval $[0, T_*)$. Instead of the semigroup
techniques used in~\cite{Dima0,OlesLena}, we apply Picard-like
approximations and a method suggested in~\cite[pp. 94, 95]{Ge},
which allows for constructing classical solutions in a scale of
Banach spaces\footnote{Further developments are known under the name
{\it Ovsjannikov's method}, see e.g.~\cite{trev}.}
$\{\mathcal{K}_\alpha\}_{\alpha \in \mathcal{I}\subset \mathbb{R}}$,
$\mathcal{K}_{\alpha'} \subset \mathcal{K}_{\alpha''}$ for $\alpha''
< \alpha'$. Namely, in Theorem~\ref{otm} we show that, for any
$\alpha_0\in \mathbb{R}$ and any $\alpha< \alpha_0$, there exists
$T(\alpha_0, \alpha)>0$ such that, for any $t\in [0,T(\alpha_0,
\alpha))$, there exists $\alpha_t \in (\alpha, \alpha_0)$ such that
the problem (\ref{R4}) with $k_0 \in \mathcal{K}_{\alpha_0}$ has a
classical solution $k_t \in \mathcal{K}_{\alpha_t}$ being the
correlation function of a certain $\mu_t$. The latter fact is
obtained by means of the corresponding result of~\cite{Nancy}. This
yields the evolution $\mu_0 \mapsto \mu_t$. In addition, in
Theorem~\ref{2tm} we show that, for $k_0 (\eta) \leq
\varkappa^{|\eta|}$, the solution obeys $k_t(\eta) \leq
\varkappa^{|\eta|}$ and hence can be continued in time to the whole
$\mathbb{R}_{+}$. These are the main results of Section 3. In
Section 4, we perform the Vlasov scaling and obtain the Vlasov
hierarchy---a linear evolution equation $(d/dt) r_t = L_V r_t$,
which we study in the same scale of Banach spaces where the
correlation functions evolve. Its main peculiarity is the fact that
if $r_0$ is the correlation function of a nonhomogeneous Poisson
measure $\pi_{\varrho_0}$ with density $\varrho_0$, then the
solution $r_t$ is the correlation function for $\pi_{\varrho_t}$
with $\varrho_t$ satisfying a nonlinear nonlocal equation, see
Lemma~\ref{Vlm} and Theorem~\ref{Vatm}. Finally, in
Theorem~\ref{3tm} we show that the rescaled correlation functions
converge in the scaling limit to the corresponding $r_t$. In Section
5, we briefly summarize and compare with each other the results of
Sections 3 and 4.

\section{The basic notions and the model}

\subsection{The notions}
All the details of the framework used in this paper can be found
in~\cite{Dima,Dima0,Dima1,Dima2,Tobi,Oles}. We consider an infinite
system of point particles moving in $\mathbb{R}^d$, $d\geq 1$.
 By $\mathcal{B}(\mathbb{R}^d)$ and
 $\mathcal{B}_{\rm b}(\mathbb{R}^d)$ we denote the set of all
Borel and the set of all bounded Borel subsets of $\mathbb{R}^d$, respectively.
 For $X\in \mathcal{B}(\mathbb{R}^d)$, the set of $n$-particle configurations in $X$ is
 \[ \Gamma^{(0)}_X = \{ \emptyset\}, \qquad \Gamma^{(n)}_X = \{\eta \subset X: |\eta| = n \}, \ \ n\in \mathbb{N}.
 \]
 where $| \cdot|$ denotes cardinality.
$\Gamma^{(n)}_X$ can be identified with the symmetrization of $\{(x_1, \dots , x_n)\in X^n: x_i \neq x_j, \ {\rm for} \ j\neq j\}$, which allows one to introduce the corresponding topology and hence the Borel $\sigma$-algebra $\mathcal{B}(\Gamma^{(n)}_X)$.
The set of finite configurations in $X$ is
\[
\Gamma_{0,X} = \bigsqcup_{n\in \mathbb{N}_0} \Gamma^{(n)}_X .
\]
We equip it with the topology of the disjoint union and hence with
the Borel $\sigma$-algebra $\mathcal{B}(\Gamma_{0,X})$. For
$X=\mathbb{R}^d$, we write $\Gamma^{(n)} $ and $\Gamma_{0}$ meaning
$\Gamma^{(n)}_X $ and $\Gamma_{0,X}$, respectively. The restriction
of the Lebesgue product measure $dx_1 dx_2 \cdots dx_n $ to
$(\Gamma^{(n)}, \mathcal{B}(\Gamma^{(n)}))$ is denoted by $m^{(n)}$.
Then the Lebesgue-Poisson measure on $(\Gamma_0,
\mathcal{B}(\Gamma_0))$ is
\begin{equation}
 \label{1A}
 \lambda = \delta_{\emptyset} + \sum_{n=1}^\infty \frac{1}{n!} m^{(n)}.
\end{equation}
For any $X\in \mathcal{B}(\mathbb{R}^d)$, the restriction of $\lambda$ to $\Gamma_{0,X}$ will also be denoted by $\lambda$.

The set of all configurations in $\mathbb{R}^d$ is
\begin{equation}
 \label{2A}
 \Gamma = \{ \gamma \subset \mathbb{R}^d : |\gamma \cap \Lambda| < \infty \ \ {\rm for} \ {\rm all} \ \ \Lambda \in \mathcal{B}_{\rm b}(\mathbb{R}^d)\}.
\end{equation}
We equip it with the vague topology---the weakest topology in which
all the maps
\[
\Gamma \ni \gamma \mapsto \langle \gamma , f \rangle = \sum_{x\in \gamma} f(x) , \quad f\in C_0 (\mathbb{R}^d),
\]
are continuous. Here $C_0 (\mathbb{R}^d)$ stands for the set of all
continuous $f:\mathbb{R}^d \rightarrow \mathbb{R}$, which have
compact supports. The vague topology on $\Gamma$ admits a
metrization, which turns it into a complete and separable metric
(Polish) space, see e.g.~\cite{Oles0}. By $\mathcal{B}(\Gamma)$ we
denote the corresponding Borel $\sigma$-algebra. It turns out that
the measurable space $(\Gamma ,\mathcal{B}(\Gamma))$ is the
projective limit of the family $\{(\Gamma_{0,\Lambda}
,\mathcal{B}(\Gamma_{0,\Lambda}))\}_{\Lambda \in \mathcal{B}_{\rm
b}(\mathbb{R}^d)}$. Then the Poisson measure $\pi$ on $(\Gamma
,\mathcal{B}(\Gamma))$ is the projective limit of the family $\{
\pi^\Lambda\}_{\Lambda \in \mathcal{B}_{\rm b}(\mathbb{R}^d)}$,
where
\begin{equation}
 \label{3A}
 \pi^\Lambda = \exp ( - m(\Lambda)) \lambda,
\end{equation}
$m(\Lambda)$ being the Lebesgue measure of $\Lambda$. The Poisson
measure $\pi_{\varrho}$ corresponding to the density $\varrho:
\mathbb{R}\rightarrow \mathbb{R}_{+}$ is introduced by means of the
measure $\lambda_{\varrho}$, defined as in (\ref{1A}) with $m$
replaced by $m_\varrho$, where, for $\Lambda \in \mathcal{B}_{\rm
b}(\mathbb{R}^d)$,
\begin{equation}
 \label{3Aa}
m_{\varrho} (\Lambda) = \int_{\Lambda} \varrho (x) d x,
\end{equation}
which is supposed to be finite. Then $\pi_{\varrho}$ is defined by its projections
\begin{equation}
 \label{3Ab}
 \pi^\Lambda_{\varrho} = \exp ( - m_{\varrho}(\Lambda)) \lambda^\Lambda_{\varrho}.
\end{equation}
For a measurable $f:\mathbb{R}^d \rightarrow \mathbb{R}$ and $\eta \in \Gamma_0$, the Lebesgue-Poisson exponent is
\begin{equation}
 \label{4A}
 e(f, \eta) = \prod_{x\in \eta} f(x), \qquad e(f, \emptyset ) = 1.
\end{equation}
Clearly $ e(f, \cdot)\in L^1 (\Gamma_0, d \lambda)$ for any $f \in L^1 (\mathbb{R}^d)$, and
\begin{equation}
 \label{5A}
\int_{\Gamma_0} e(f, \eta) \lambda (d\eta) =
\exp\left\{\int_{\mathbb{R}^d} f(x) d x \right\}.
\end{equation}
A set $M\in \mathcal{B}(\Gamma_0)$ is said to be bounded if
\begin{equation}
 \label{6A}
 M \subset \bigsqcup_{n=0}^N \Gamma^{(n)}_\Lambda
\end{equation}
for some $\Lambda \in \mathcal{B}_{\rm b}(\mathbb{R}^d)$ and $N\in
\mathbb{N}$. By $B_{\rm bs} (\Gamma_0)$ we denote the set of all
bounded measurable functions $G: \Gamma_0 \rightarrow \mathbb{R}$,
which have bounded supports. That is, each such $G$ is the zero
function on $\Gamma_0 \setminus M$ for some bounded $M$. Noteworthy,
any measurable $G: \Gamma_0 \rightarrow \mathbb{R}$ is in fact a
sequence of measurable symmetric functions $G^{(n)} :
(\mathbb{R}^d)^n \rightarrow \mathbb{R}$.

For $\Lambda \in \mathcal{B}_{\rm b}(\mathbb{R}^d)$ and $\gamma\in
\Gamma$, by $\gamma_\Lambda$ we denote $\gamma \cap \Lambda$; thus,
$\gamma_\Lambda \in \Gamma_{0,\Lambda}$. A measurable function
$F:\Gamma \rightarrow \mathbb{R}$ is called a {\it cylinder}
function if there exist $\Lambda \in \mathcal{B}_{\rm
b}(\mathbb{R}^d)$ and a measurable $G: \Gamma_{0,\Lambda}\rightarrow
\mathbb{R}$ such that $F(\gamma) = G(\gamma_\Lambda)$ for all
$\gamma \in \Gamma$. By $\mathcal{F}_{\rm cyl}(\Gamma)$ we denote
the set of all cylinder functions. For $\gamma \in \Gamma$, by
writing $\eta \Subset \gamma$ we mean that $\eta \subset \gamma$ and
$\eta$ is finite, i.e., $\eta \in \Gamma_0$. For $G \in B_{\rm
bs}(\Gamma_0)$, we set
\begin{equation}
 \label{7A}
 (KG)(\gamma) = \sum_{\eta \Subset \gamma} G(\eta), \qquad \gamma \in \Gamma.
\end{equation}
It is known that $K$ is linear and positivity preserving, and maps
$B_{\rm bs}(\Gamma_0)$ into $\mathcal{F}_{\rm cyl}(\Gamma)$ (see
e.g. \cite{Tobi}).

By $\mathcal{M}^1_{\rm fm} (\Gamma)$ we denote the set of all probability measures on $(\Gamma, \mathcal{B}(\Gamma))$ which have finite local moments, that is, for which
\begin{equation}
 \label{8A}
 \int_\Gamma |\gamma_\Lambda|^n \mu(d \gamma) < \infty \quad \ \ \ {\rm for} \ \ {\rm all} \ \ n\in \mathbb{N} \ \ {\rm and}\ \ \Lambda \in \mathcal{B}_{\rm b} (\mathbb{R}^d).
\end{equation}
A measure $\rho$ on $(\Gamma_0, \mathcal{B}(\Gamma_0))$ is said to be {\it locally finite} if $\rho(M)< \infty$ for every bounded $M\subset \Gamma_0$. By $\mathcal{M}_{\rm lf} (\Gamma_0)$ we denote the set of all such measures.
For $\Lambda \in \mathcal{B}_{\rm b} (\mathbb{R}^d)$, by $p_\Lambda$ we denote the map $\Gamma \ni \gamma \mapsto p_\Lambda (\gamma) = \gamma_\Lambda$. Then, for $A\subset \Gamma_{0,\Lambda}$, we write $p^{-1}_\Lambda (A) = \{ \gamma \in \Gamma : p_\Lambda (\gamma ) \in A\}$.
A measure $\mu \in \mathcal{M}^1_{\rm fm} (\Gamma)$
is said to be {\it locally absolutely continuous} with respect to the Poisson measure $\pi$ if, for every
$\Lambda \in \mathcal{B}_{\rm b} (\mathbb{R}^d)$, $\mu^\Lambda := \mu \circ p_\Lambda^{-1}$ is absolutely continuous with respect to $\pi^\Lambda$, see (\ref{3A}).

Let $M\subset \Gamma_0$ be bounded, and let $\mathbb{I}_M$ be its
indicator function on $\Gamma_0$. Then $\mathbb{I}_M$ is in $B_{\rm
bs}(\Gamma_0)$ and hence one can apply (\ref{7A}). For $\mu \in
\mathcal{M}^1_{\rm fm} (\Gamma)$, let
\begin{equation}
 \label{9A}
 \rho_\mu (M) = \int_{\Gamma} (K\mathbb{I}_M) (\gamma) \mu(d \gamma),
\end{equation}
which uniquely determines a measure $\rho_\mu \in \mathcal{M}_{\rm
lf}(\Gamma_0))$. It is called the {\it correlation measure} for
$\mu$. This defines the map $K^*: \mathcal{M}^1_{\rm fm} (\Gamma)
\rightarrow \mathcal{M}_{\rm lf}(\Gamma_0))$ such that $K^*\mu =
\rho_\mu$. In particular, $K^*\pi = \lambda$. It is known that,
see~\cite[Proposition 4.14]{Tobi}, $\rho_\mu $ is absolutely
continuous with respect to $\lambda$ if $\mu$ is locally absolutely
continuous with respect to $\pi$. In this case, we have that for any
$\Lambda \in \mathcal{B}_{\rm b} (\mathbb{R}^d)$,
\begin{equation}
 \label{9AA}
 k_\mu (\eta) = \frac{d \rho_\mu}{d \lambda}(\eta) = \int_{\Gamma_{0,\Lambda}} \frac{d\mu^\Lambda}{d \pi^\Lambda} (\eta \cup \gamma) \pi^\Lambda (d \gamma).
\end{equation}
The Radon--Nikodym derivative $k_\mu$ is called the {\it correlation
function} corresponding to the measure $\mu$.

Finally, we mention the following integration rule, c.f.~\cite[Lemma
2.1]{Dima0},
\begin{equation}
 \label{12A}
\int_{\Gamma_0} \sum_{\xi \subset \eta} H(\xi, \eta \setminus \xi,
\eta) \lambda (d \eta) = \int_{\Gamma_0}\int_{\Gamma_0}H(\xi, \eta,
\eta\cup \xi) \lambda (d \xi) \lambda (d \eta),
\end{equation}
 which holds for any appropriate function $H$.

\subsection{The model}

Let $\phi: \mathbb{R}^d \rightarrow \mathbb{R}_{+} := [0, +\infty)$
be such that $\phi(x ) = \phi (-x)$ and be integrable in the
following sense
\begin{equation}
 \label{13A}
 c_\phi := \int_{\mathbb{R}^d} \left( 1 - e^{-\phi(x)}\right)d x < \infty.
\end{equation}
For $\gamma\in \Gamma$, we set
\begin{equation}
 \label{14A}
 E^\phi (x, \gamma) = \sum_{y\in \gamma} \phi(x-y) ,
\end{equation}
with the possibility that $E^\phi (x, \gamma) = +\infty$ for some
$\gamma$. In the model we consider, the dynamics of the observables
is defined by the `generator' (\ref{R20}) with $\phi$ just mentioned
and $\varkappa >0$ being the birth {\it activity} parameter. The
action of the `generator' (\ref{R20}) on $F\in \mathcal{F}_{\rm
cyl}(\Gamma)$ is well-defined. Indeed, for any
 $F\in \mathcal{F}_{\rm cyl}(\Gamma)$, one finds $\Lambda \in \mathcal{B}_{\rm b} (\mathbb{R}^d)$
such that $F(\gamma \setminus x) = F(\gamma \cup x) = F(\gamma)$ for any $x\in \Lambda^c : = \mathbb{R}^d \setminus \Lambda$.
Thus, the sum and the integral in (\ref{R20}) are finite.

Following the general scheme developed in~\cite{Oles1} one
constructs the evolution of the {\it quasi-observables}, which are
functions on $\Gamma_0$. This evolution is obtained as a solution to
the following Cauchy problem
\begin{equation}
 \label{16A}
\frac{d G_t}{dt} = \hat{L} G_t, \qquad G_t|_{t=0} = G_0,
\end{equation}
where $\hat{L} = K^{-1} L K$ is the so called {\it symbol} of $L$, which has the form
\begin{equation}
 \label{17A}
 (\hat{L}G) (\eta) = - |\eta| G(\eta) + \varkappa \sum_{\xi \subset \eta}\int_{\mathbb{R}^d} e(t_x , \eta\setminus \xi)
 e(\tau_x, \xi) G(\xi \cup x) d x ,
\end{equation}
where $e(t_x, \cdot)$ is defined in (\ref{4A}), and
\begin{equation}
 \label{18A}
 \tau_x (y) = e^{-\phi (x-y)} , \qquad t_x (y) = \tau_x(y) - 1.
\end{equation}
Clearly, the action of $\hat{L}$ on $G\in B_{\rm bs}(\Gamma_0)$ is well-defined. Its extension
to wider classes of $G$ will be done in a while.

For a measurable function $k:\Gamma_0 \rightarrow \mathbb{R}$ and
$G \in B_{\rm bs} (\Gamma_0)$, we define
\begin{equation}
 \label{19A}
 \langle \! \langle G,k\rangle \! \rangle = \int_{\Gamma_0} G(\eta) k(\eta)\lambda (d \eta) .
\end{equation}
This pairing can be extended to the corresponding classes of $G$ and $k$.
Then (\ref{16A}) and (\ref{19A}) lead to the following (dual) Cauchy problem
\begin{equation}
 \label{20A}
\frac{d k_t }{ d t} = L^\Delta k_t , \qquad k_t|_{t=0} = k_0.
\end{equation}
The action of $L^\Delta$ is obtained by means of (\ref{12A}) from
\[
\langle \! \langle \hat{L} G, k \rangle \! \rangle = \langle \!
\langle G, L^\Delta k \rangle \! \rangle,
\]
and from (\ref{19A}) and (\ref{17A}). It thus has the form
\begin{align}
 \label{21A}
 (L^\Delta k)(\eta)  = & - |\eta| k(\eta)\\
  &+  \varkappa \sum_{x \in \eta} e(\tau_x, \eta \setminus x) \int_{\Gamma_0} e(t_x , \xi)
  k(\eta \setminus x \cup \xi) \lambda
  (d\xi).\nonumber
\end{align}
Of course, the case of a special interest in (\ref{20A}) is where
$k_0$ is the correlation function of a certain $\mu_0\in
\mathcal{M}^1_{\rm fm} (\Gamma)$, see (\ref{9AA}). However, the mere
existence of the solution $k_t$ does not guarantee that this $k_t$
is a correlation function.

 In~\cite{Oles1,Oles,OlesLena},
the solution $G_0 \mapsto G_t$ of (\ref{16A}), for all $t\geq 0$ and
`small' $\varkappa$ and $c_\phi$, was obtained in a certain Banach
space by means of the construction of a $C_0$-semigroup based on
perturbation methods. Then the evolution of the correlation
functions $k_0 \mapsto k_t$ was obtained in the weak sense, in which
$k_t$ is defined by $k_0$ via the relation
\begin{equation}
\label{1D}
 \langle \! \langle G_0, k_t \rangle \! \rangle = \langle \! \langle G_t, k_0 \rangle \! \rangle.
\end{equation}
Regarding the problems (\ref{16A}) and (\ref{20A}), in the present article we realize the following program:
\begin{itemize}
 \item Show that (\ref{16A}) has a unique classical solution for all $\varkappa>0$ and $c_\phi$, which we do in Theorem~\ref{1tm}
for $t$ belonging to a bounded interval.
\item Show that the solution of (\ref{16A}) exists for all $t\geq 0$ if $\varkappa c_\phi < 1/e$,
which we do in Theorem~\ref{20tm}.
\item Show that (\ref{20A}) has a unique classical solution $k_t$ for all $\varkappa>0$ and $c_\phi$, being
the correlation function of a certain $\mu_t \in \mathcal{M}^1_{\rm
fm}(\Gamma)$, which yields the evolution of states $\mu_0 \mapsto
\mu_t$. We do this in Theorem~\ref{otm} for $t$ belonging to a
bounded interval.
\item Show that the solution of (\ref{16A}) exists for all $t\geq 0$ if $k_0 (\eta) \leq \varkappa^{|\eta|}$, which we do
in Theorem~\ref{2tm}.
\end{itemize}
These results give the microscopic evolution of states corresponding
to (\ref{R20}). A~similar program concerning the mesoscopic
evolution is formulated and realized in Section 4 below.

\section{The microscopic description}

\subsection{The evolution of quasi-observables}

First we study the problem (\ref{16A}), (\ref{17A}). For $\alpha \in\mathbb{R}$, we consider the Banach space
\begin{equation}
 \label{BS}
\mathcal{G}_\alpha = L^1 (\Gamma_0, e^{-\alpha |\cdot|}d \lambda),
\end{equation}
that is, $G\in \mathcal{G}_\alpha$ if
\begin{equation}
 \label{6}
\|G\|_\alpha \ \stackrel{\rm def}{=} \ \int_{\Gamma_0} \exp(-\alpha |\eta|) \left\vert G(\eta) \right\vert \lambda (d \eta) < \infty.
\end{equation}
We will seek the solution of (\ref{16A}), (\ref{17A})
as the limit of $\{G^{(n)}_t\}_{n\in \mathbb{N}_0 }\subset \mathcal{G}_\alpha$, where $G^{(0)}_t = G_0$ and
\begin{equation}
 \label{8}
G^{(n)}_t = G_0 + \int_0^t \hat{L} G^{(n-1)}_s d s, \quad n \in \mathbb{N}.
\end{equation}
The latter can be iterated to give
\begin{equation}
 \label{8C}
G^{(n)}_t = G_0 + \sum_{m=1}^n \frac{1}{m!}t^m
\hat{L}^m G_0.
\end{equation}
We have $\tau_x(y) \leq 1$ since $\phi \geq 0$, see (\ref{18A});
hence, from (\ref{17A})
\begin{align*}
|\hat{L} G (\eta)| & \leq  |\eta| |G(\eta)| + \varkappa \sum_{\xi
\subset \eta}
 \int_{\mathbb{R}^d} e(|t_x|, \eta \setminus \xi) |G(\xi \cup x) | d x \\
  &:= H_1 (\eta) + H_2 (\eta).
\end{align*}
For any $\alpha'$, $\alpha''$ such that $ \alpha' < \alpha'' $, we
have
\begin{align}
 \label{9y}
 \|H_1 \|_{\alpha''} & =  \int_{\Gamma_0} |\eta| \exp\left(- (\alpha'' - \alpha')|\eta| \right) |G(\eta)| \exp\left(- \alpha' |\eta|\right) \lambda (d\eta)
 \\
& \leq \frac{\|G\|_{\alpha'}}{(\alpha'' - \alpha')e}, \nonumber
\end{align}
where we have used the following obvious estimate
\[
 |\xi| \exp(- (\alpha'' - \alpha') |\xi|) \leq \frac{1}{ (\alpha'' - \alpha')e}.
\]
Furthermore,
\begin{align*}
 \|H_2 \|_{\alpha''} &= \varkappa \int_{\Gamma_0} \sum_{\xi \subset \eta}\int_{\mathbb{R}^d} e(|t_x|; \eta \setminus
 \xi)
\exp(- \alpha'' |\eta \setminus \xi| - \alpha'' |\xi|)|G(\xi \cup x)
| d x \lambda (d\eta)\\& = \varkappa \int_{\Gamma_0}\int_{\Gamma_0}
\int_{\mathbb{R}^d} e(|t_x|; \eta) \exp(- \alpha'' |\eta | -
\alpha'' |\xi|)|G(\xi \cup x) | d x \lambda (d\xi) \lambda (d\eta)
\\& \leq \varkappa
\int_{\mathbb{R}^d} \int_{\Gamma_0}\exp( - \alpha'' |\xi|) |G(\xi
\cup x) |d x \lambda (d \xi) \\& \qquad \times \sup_{y\in\mathbb{R}}
\left\{ \int_{\Gamma_0} e(|t_y|; \eta)\exp(- \alpha'' |\eta |)
\lambda (d \eta) \right\} \\& = \varkappa \exp\left(c_\phi
e^{-\alpha''} \right) \int_{\Gamma_0} \sum_{x\in \xi} \exp\left( -
\alpha'' |\xi\setminus x|\right) |G(\xi)| \lambda ( d\xi)\\
&=\varkappa \exp\left(\alpha'' + c_\phi e^{-\alpha''} \right)
\int_{\Gamma_0} | \xi| \exp\left( - \alpha'' |\xi\setminus x|\right)
|G(\xi)| \lambda ( d\xi)\\& = \varkappa \exp\left(\alpha'' + c_\phi
e^{-\alpha''} \right) \frac{\|G\|_{\alpha'}}{(\alpha'' - \alpha')e},
\end{align*}
where we have used also (\ref{5A}) and (\ref{12A}). By means of the
latter estimate and (\ref{9y}) we finally get
\begin{equation}
 \label{9}
 \|\hat{L}G\|_{\alpha''} \leq \frac{\|G\|_{\alpha'}}{(\alpha'' - \alpha')e} \left[1 +
 \varkappa \exp\left(\alpha''+ c_\phi e^{-\alpha''} \right) \right].
\end{equation}
From (\ref{9}) we see that $\hat{L}$ can be defined as a bounded linear operator
$\hat{L}: \mathcal{G}_{\alpha'} \rightarrow
\mathcal{G}_{\alpha''}$ with the norm
\begin{equation}
 \label{10}
\|\hat{L}\|_{\alpha' \alpha''} \leq \frac{ 1}{(\alpha'' - \alpha')e}
\left[1 + \varkappa \exp\left(\alpha'' + c_\phi e^{-\alpha''}
\right) \right].
\end{equation}
Given $\alpha_0\in \mathbb{R}$, $\alpha> \alpha_0$, $m\in
\mathbb{N}$, and $l = 0, \dots , m$, we take $\alpha_l = \alpha_0 +
l \epsilon$, $\epsilon = (\alpha - \alpha_0)/m$. Then by (\ref{10})
we get
\begin{align}
 \label{10C}
 \| \hat{L}^m\|_{ \alpha_0 \alpha} & \leq \|\hat{L}\|_{\alpha_0\alpha_1} \cdots \|\hat{L}\|_{\alpha_{m-1}\alpha}
 \leq (m M)^m, \\
 \intertext{where} M &= \frac{1}{(\alpha - \alpha_0)e} \left[1 + \varkappa \exp\left(\alpha+ c_\phi e^{-\alpha_0} \right) \right]. \nonumber
\\\intertext{Put}
 \label{7}
T(\alpha, \alpha_0) &= \frac{\alpha - \alpha_0}{1 + \varkappa
\exp\left(\alpha + c_\phi e^{-\alpha_0} \right)}.
\end{align}
Note that
\begin{align}
\label{7R} T(\alpha, \alpha_0) & <  \frac{1}{\varkappa} \exp\left(
\log (\alpha-\alpha_0) - (\alpha - \alpha_0) - \alpha_0 - c_\phi
e^{-\alpha_0} \right)\\& \leq \frac{1}{\varkappa} \exp\left( -1 -
\log c_\phi - 1 \right) = \frac{1}{ e^2 \varkappa c_\phi}. \nonumber
\end{align}
Hence, we can set
\begin{equation}
 \label{7Q}
T_* {:=} \sup_{\alpha \in \mathbb{R}} \sup_{\alpha_0<
\alpha}T(\alpha, \alpha_0)<\infty.
\end{equation}
Our main result concerning the problem (\ref{16A}), (\ref{17A}) is
the following statement.
\begin{theorem}
 \label{1tm}
Let $\alpha_0$ and $\alpha$ be any real numbers such that $\alpha_0 < \alpha$.
Then the problem (\ref{16A}), (\ref{17A}) with $G_0 \in \mathcal{G}_{\alpha_0}$
has a unique classical solution $G_t \in \mathcal{G}_\alpha$ on the time interval $t \in [0, T(\alpha, \alpha_0))$.
\end{theorem}
\begin{proof}
Applying (\ref{10C}) in (\ref{8C}) we get that the sequence
$\{G^{(n)}_t\}_{n\in \mathbb{N}_0}$ converges in
$\mathcal{G}_\alpha$ uniformly on any $[0,T] \subset [0, T (\alpha,
\alpha_0))$. In fact, one can show that for any
$\alpha_1\in(\alpha_0,\alpha)$ this sequence converges in
$\mathcal{G}_{\alpha_1}$ to some $\tilde{G}_t$ uniformly  on any
$[0,T] \subset [0, T (\alpha_1, \alpha_0))$. Note that $T (\alpha,
\alpha_0)$ continuously depend on $\alpha$, therefore, we may
consider any $[0,T] \subset [0, T (\alpha, \alpha_0))$. Since
$\|\cdot\|_{\alpha}\leq\|\cdot\|_{\alpha'}$, $\tilde{G}_t$ coincide
with $G_t$ in $\mathcal{G}_{\alpha}$. Hence, $G_t$ belongs to any
$\mathcal{G}_{\alpha_1}$, in particular, $G_t$ is in the domain of
the (unbounded) operator $\hat{L}$ in the space
$\mathcal{G}_{\alpha}$. By \eqref{10C},  the sequence
$\{\frac{d}{dt}G^{(n)}_t\}_{n\in \mathbb{N}_0}$ also converges in
$\mathcal{G}_\alpha$ uniformly on any $[0,T] \subset [0, T (\alpha,
\alpha_0))$. Therefore, $G_t$ is a solution of (\ref{16A}),
(\ref{17A}). The uniqueness can be shown in the same way as in
\cite[p.~16,~17]{trev}.
\end{proof}

The statement just proven describes systems with any $\varkappa$ and
$c_\phi$. It is, however, possible to get more if one imposes
appropriate restrictions on these parameters. Namely, the dynamics
in this case is described by a $C_0$-semigroup
$S(t):\mathcal{G}_\alpha \rightarrow \mathcal{G}_\alpha$, where the
space is the same as in (\ref{BS}). Set
\begin{equation}
 \label{j}
\mathcal{G}_\alpha^+ = \{ G \in \mathcal{G}_\alpha \ : \ G \geq 0\}, \quad
\mathcal{H}_\alpha = \{ G \in \mathcal{G}_\alpha \ : \ |\cdot | G \in \mathcal{G}_\alpha\},
\end{equation}
and also
\begin{equation}
 \label{j0}
\mathcal{H}_\alpha^+ = \mathcal{H}_\alpha \cap \mathcal{G}_\alpha^+.
\end{equation}
\begin{theorem}
 \label{20tm}
Assume that
\begin{equation}
 \label{jj}
\varkappa c_\phi < 1/e.
\end{equation}
For such $\phi$, let $\alpha_\phi = \ln c_\phi$. Then, for every
$G_0 \in \mathcal{H}_{\alpha_\phi}$, the problem (\ref{16A}),
(\ref{17A}) has a unique classical solution $G_t\in
\mathcal{G}_{\alpha_\phi}$, $t\geq 0$, given by $G_t = S(t) G_{0}$,
where $\{S(t)\}_{t\geq 0}$ is a $C_0$-semigroup on
$\mathcal{G}_{\alpha_\phi}$.
\end{theorem}
\begin{remark} \label{1rk}
(a) The above theorem is true not only for $\alpha = \alpha_\phi$
but also for $\alpha \in (\alpha_\phi - \delta , \alpha_\phi +
\delta)$ for some $\delta >0$, see the proof below. (b) The
condition (\ref{jj}) is well-known, see~\cite[Chapter 4]{Ruelle}. It
provides the convergence of cluster expansions for the gas of
classical particles with the pair-wise repulsion $U(x,y) =
\phi(x-y)$.
\end{remark}
In the proof of Theorem~\ref{20tm} we use two statements. The first
one is an adaptation of~\cite[Corollary 5.16]{[BA]}.
\begin{proposition} \label{BAprop}
Given $\alpha \in \mathbb{R}$, assume that $A$ is the generator of a positive $C_0$-semigroup
in $\mathcal{G}_\alpha$ and $B= B_1 - B_2$ be such that $A + B_1 + B_2$
is also the generator of a positive $C_0$-semigroup. Then $A +B$ generates a $C_0$-semigroup in $\mathcal{G}_\alpha$.
\end{proposition}
For $\eta\in \Gamma_0$, we set $\Xi_e(\eta) = \{ \xi \subset \eta \
: \ |\eta \setminus \xi|\text{ is even}\}$ and $\Xi_o(\eta) = \{ \xi
\subset \eta \ : \ |\eta \setminus \xi|\text{ is odd}\}$, and
thereby
\begin{align}
\label{j1} (B_+ G ) (\eta)& =  \varkappa \sum_{\xi \in \Xi_e(\eta)}
\int_{\mathbb{R}^d} e(\tau_x, \xi)
 e(t_x, \eta \setminus \xi) G(\xi\cup x) d x, \\\label{j111}
(B_{-} G ) (\eta)& = - \varkappa \sum_{\xi \in \Xi_0(\eta)}
 \int_{\mathbb{R}^d} e(\tau_x, \xi) e(t_x, \eta \setminus \xi) G( \xi\cup x) d x.
\end{align}
We also set
\begin{equation}
 \label{j2}
(A G ) (\eta) = - |\eta| G(\eta),
\end{equation}
and
\begin{align}
 \label{j3}
(\hat{L}^+ G)(\eta) & = ((A + B_{+} + B_{-})G)(\eta)\\ & = - |\eta|
G(\eta) + \varkappa \sum_{\xi\subset \eta} \int_{\mathbb{R}^d}
e(\tau_x; \xi) e(|t_x|; \eta \setminus \xi) G( \xi\cup x) d x.
\nonumber
\end{align}
Clearly, for any $\alpha$, $A$ with $Dom (A) = \mathcal{H}_\alpha$
generates a positive semigroup of contractions in
$\mathcal{G}_\alpha$. All the three operators $-A$, $B_{\pm}$ are
positive. The second statement we need to prove Theorem~\ref{20tm}
is an adaptation to our case of~\cite[Theorem 2.7]{TV}.
\begin{proposition}
 \label{TVpn}
Let $A$ and $B_{\pm}$ be as above. Suppose that there exist real $\alpha$ and $\alpha'$, $\alpha' < \alpha$, such that
\begin{align}
 \label{j4}
& \forall G\in \mathcal{H}_\alpha^+: &&\int_{\Gamma_0} (\hat{L}^+
G)(\eta) \exp(-\alpha |\eta|) \lambda (d \eta) \leq 0,
 \\ \label{j5}
& \forall G\in \mathcal{H}_{\alpha'}^+: &&\int_{\Gamma_0} (\hat{L}^+
G)(\eta) \exp(-\alpha' |\eta|)
 \lambda (d \eta)\\&&&\leq
 C \int_{\Gamma_0} G(\eta) \exp(-\alpha' |\eta|) \lambda (d \eta)  - \varepsilon \int_{\Gamma_0} |\eta| G(\eta) \exp(-\alpha |\eta|) \lambda(d\eta) \nonumber
\end{align}
for some positive $C$ and $\varepsilon$. Then the closure of
$\hat{L}^+$ in $\mathcal{G}_\alpha$ generates a positive
$C_0$-semigroup in this space.
\end{proposition}
\begin{proof}[Proof of Theorem~\ref{20tm}.] Given $\alpha$, for $G\in
\mathcal{H}_\alpha^+$, similarly as in (\ref{10}) we get
\begin{align}
 \label{j6}
& \int_{\Gamma_0} (\hat{L}^+ G)(\eta) \exp(-\alpha |\eta|)\lambda (d
\eta) \\
 \leq & - \left(1 - \varkappa \exp\left( \alpha + c_\phi e^{-\alpha}\right) \right)
\int_{\Gamma_0} |\eta| G(|\eta|)\exp(-\alpha |\eta|) \lambda (d \eta).\nonumber
\end{align}
By (\ref{jj}), we have $\varkappa \exp\left( \alpha_\phi + c_\phi
e^{-\alpha_\phi}\right) < 1$; hence, one can pick small enough
$\delta>0$ such that $\varkappa \exp\left( \alpha + c_\phi
e^{-\alpha}\right) < 1$ for any $\alpha \in (\alpha_\phi - \delta,
\alpha_\phi + \delta)$. Then we fix such $\alpha$ and obtain
(\ref{j4}) as the coefficient in the last line in (\ref{j6}) is
negative. Next we pick $\sigma>0$ such that $\alpha' := \alpha -
\sigma$ is also in $(\alpha_\phi - \delta, \alpha_\phi + \delta)$.
For this $\alpha'$, ${\rm LHS(\ref{j5})} \leq 0$ for $G\in
\mathcal{H}^{+}_{\alpha'}$. On the other hand,
\begin{align*}
 \int_{\Gamma_0} |\eta| G(|\eta|)\exp(-\alpha |\eta|) \lambda (d \eta) & =  \int_{\Gamma_0} |\eta| e^{- \sigma |\eta|}
G(|\eta|)\exp(-\alpha' |\eta|) \lambda (d \eta)\\ & \leq
\frac{1}{e\sigma} \int_{\Gamma_0} G(|\eta|)\exp(-\alpha' |\eta|)
\lambda (d \eta),
\end{align*}
which yields that ${\rm RHS(\ref{j5})} \geq 0$ for any $C$ and
sufficiently small $\varepsilon$. Hence, we can apply
Proposition~\ref{TVpn}, by which the closure of $\hat{L}^+$ in
$\mathcal{G}_\alpha$ generates a positive $C_0$-semigroup. But,
under the condition (\ref{jj}), $\hat{L}^+$ is closed as $A$ is
closed. Thus, we are able now to apply Proposition~\ref{BAprop} and
complete the proof.
\end{proof}

\subsection{The evolution of correlation functions and states}

In this subsection,
the problem (\ref{20A}), (\ref{21A}) will be studied in the Banach space, c.f. (\ref{BS}),
\begin{equation}
 \label{2D}
 \mathcal{K}_\alpha = L^\infty (\Gamma_0, e^{\alpha |\cdot|} d \lambda)
\end{equation}
which we equip with the norm
\begin{equation}
 \label{o5}
\|k\|_\alpha = {\rm ess \ sup}\left\{ |k(\eta)| \exp(\alpha |\eta|): \eta \in \Gamma_0\right\}.
\end{equation}
For the sake of convenience, here we use the same notation $\|\cdot\|_\alpha$ as in (\ref{6}), which, however,
should not cause any ambiguity as it will always be clear from the context which norm is meant.

Recall that $\mathcal{M}^1_{\rm fm} (\Gamma)$ stands for the set of probability measures on $\Gamma$ obeying
(\ref{8A}). Given $\alpha\in \mathbb{R}$, by $\mathcal{M}_\alpha$ we denote the set of all those $\mu\in\mathcal{M}^1_{\rm fm} (\Gamma)$
whose correlation functions, defined in (\ref{9AA}), belong to $\mathcal{K}_\alpha$.
The main result of this section
is contained in the following statement.
\begin{theorem}
 \label{otm}
Fix any $\alpha_0\in \mathbb{R}$ and any $\alpha < \alpha_0$, and let $T(\alpha_0, \alpha)$ be as in (\ref{7}).
Then, for every $t\in (0, T(\alpha_0, \alpha))$, there exists $\alpha_t \in (\alpha, \alpha_0)$ such that
the problem (\ref{20A})
 with $k_0\in \mathcal{K}_{\alpha_0}$
has a unique classical solution $k_t \in \mathcal{K}_{\alpha_t}$.
This solution is the correlation function of a certain (unique) $\mu_t \in \mathcal{M}_{\alpha_t}$, which yields the evolution of states
$\mu_0 \mapsto \mu_t$ of the considered model in the scale $\{\mathcal{M}_{\alpha_t}\}_{
t\in [0, T(\alpha_0, \alpha))}$.
\end{theorem}
\begin{proof}
First we prove the existence of the solution $k_t$ in the Banach space $\mathcal{K}_{\alpha_t}$ and then
show that it is a correlation function.

As in (\ref{8}), we seek the solution of (\ref{20A}) as the limit of the sequence $\{k^{(n)}_t \}_{n\in \mathbb{N}_0}$, where
\begin{equation}
 \label{3D}
k_t^{(n)} = k_0 + \int_0^t L^\Delta k^{(n-1)}_s d s, \quad k_t^{(0)}
= k_0,
\end{equation}
which yields, c.f. (\ref{8C}),
\begin{equation}
 \label{4D}
 k_t^{(n)} = k_0 + \sum_{m=1}^n \frac{1}{m!} t^m \left( L^\Delta \right)^m k_0.
\end{equation}
By (\ref{o5}),
\[
|k(\eta)| \leq \|k\|_\alpha \exp( - \alpha |\eta|).
\]
Thus, for $\alpha'$ and $\alpha''$ as in (\ref{9}), we get from (\ref{21A})
\begin{align*}
 \left\vert (L^\Delta k)(\eta)\right\vert \leq &\|k\|_{\alpha''}
\exp( - \alpha' |\eta|) \bigg{\{}|\eta| \exp( - (\alpha'' - \alpha')
|\eta|)\\&
 + \varkappa \exp\left( \alpha'' - (\alpha'' - \alpha')|\eta|\right) \sum_{x\in \eta}
 \int_{\Gamma_0} e(|t_x|, \xi) \exp( - \alpha'' |\xi|) \lambda (d\xi)
\bigg{\}}.
\end{align*}
Similarly as in producing (\ref{9}) we then get from the latter
\begin{equation}
 \label{o6}
\|L^\Delta k\|_{\alpha'} \leq \frac{\|k\|_{\alpha''}}
{(\alpha''-\alpha')e}\left[1 + \varkappa \exp\left(\alpha''+ c_\phi
e^{-\alpha''} \right) \right].
\end{equation}
Hence, $L^\Delta $ can be defined as a bounded linear operator
$L^\Delta :\mathcal{K}_{\alpha''} \rightarrow \mathcal{K}_{\alpha'}$ with the norm
\begin{equation}
 \label{o7}
\|L^\Delta \|_{\alpha''\alpha'} \leq \frac{1}{(\alpha'' -\alpha')e}
 \left[1 + \varkappa \exp\left(\alpha_0+ c_\phi e^{-\alpha} \right) \right].
\end{equation}
Now we fix $t\in (0, T(\alpha_0, \alpha))$, and for $\tilde{\alpha}
\in (\alpha, \alpha_0)$ and a given $m\in \mathbb{N}$, set $\alpha_l
= \alpha_0 - l \epsilon$, $\epsilon = (\alpha_0 -
\tilde{\alpha})/m$. Then, by (\ref{o7}),
\begin{align}
 \label{5D}
 \| \left( L^\Delta \right)^m \|_{\alpha_0 \tilde{\alpha}}& \leq
m^m \left[\frac{1 + \varkappa \exp\left(\alpha_0+ c_\phi e^{-\alpha} \right)}{(\alpha_0 - \tilde{\alpha})e} \right]^m\\
&= \left( \frac{m}{e}\right)^m \left(\frac{\alpha_0 -
\alpha}{\alpha_0 - \tilde{\alpha}} \right)^m \frac{1}{[T(\alpha_0,
\alpha)]^m}, \nonumber
\end{align}
see (\ref{7}). Pick $\delta>0$ such that
$t+\delta<T(\alpha_0,\alpha)$. Then for
\begin{equation}
 \label{alpha}
\alpha<\alpha_t \leq \alpha_0 \left( 1 - \frac{t+\delta}{T(\alpha_0,
\alpha)}\right) + \alpha \frac{t+\delta}{T(\alpha_0, \alpha)},
\end{equation}
we have from (\ref{5D})
\[
 \| \left( L^\Delta \right)^m \|_{\alpha_0 {\alpha}_t} \leq \left( \frac{m}{(t+\delta)e}\right)^m.
\]
Applying this estimate in (\ref{3D}) we obtain the convergence of
the sequence $\{k_t^{(n)}\}_{n\in \mathbb{N}_0}$ in
$\mathcal{K}_{\alpha_t}$, which yields the existence and uniqueness
of the solution $k_t$ similarly to that in the proof of
Theorem~\ref{1tm}.

Now let us show that the solution just constructed is such that, for any $t\in (0, T(\alpha_0, \alpha))$,
there exists
$\mu_t \in \mathcal{M}^1_{\rm fm} (\Gamma)$ such that, c.f. (\ref{9AA}),
\begin{equation}
 \label{8D}
 k_t (\eta) = \frac{d (K^* \mu_t)}{d \lambda} (\eta),
\end{equation}
if $k_0\in \mathcal{K}_{\alpha_0}$ is the correlation function of
the corresponding $\mu_0 \in \mathcal{M}^1_{\rm fm} (\Gamma)$. Under
the condition (\ref{13A}), there exists a proper $\mathcal{S}\subset
\Gamma$ and an $\mathcal{S}$-valued process with sample paths in the
Skorokhod space $D_\mathcal{S}(\mathbb{R}_+)$ associated with $L$,
see~\cite[Theorem 2.13]{Nancy}. This yields the evolution $\mu_0
\mapsto \mu_t$, where $\mu_t$ is the law of the process,
 and hence the evolution of the corresponding correlation functions
$k_{\mu_0} \mapsto k_{\mu_t}$, which satisfy (\ref{20A}). By the
uniqueness just established, $k_t = k_{\mu_t}$ for all $t \in [0,
T(\alpha_0, \alpha))$, which yields (\ref{8D}) and hence completes
the proof.
\end{proof}
\begin{remark}
 \label{Urk}
Theorem~\ref{otm} establishes the evolution $k_0 \mapsto k_t$ which
takes places in the scale of spaces $\mathcal{K}_{\alpha_t}$. It is
clear from the proof that all such spaces are contained in
$\mathcal{K}_\alpha$, that is, the mentioned theorem can be
formulated similarly as Theorem~\ref{1tm}.
\end{remark}
Another our remark addresses the regularity of the solutions $k_t$.
Instead of (\ref{2D}) let us consider
\begin{equation}
 \label{6D}
\widetilde{\mathcal{K}}_\alpha = \{k \in C (\Gamma_0\rightarrow \mathbb{R}): \|k\|_\alpha < \infty\},
\end{equation}
where this time
\begin{equation}
 \label{7D}
 \|k\|_\alpha = {\sup}\left\{ |k(\eta)| \exp(\alpha |\eta|): \eta \in \Gamma_0\right\}.
\end{equation}
\begin{remark}
 \label{o1rk}
Let $\alpha_0$, $\alpha$, and $T(\alpha_0,\alpha)$ be as in
Theorem~\ref{otm}. Suppose in addition that the function $\phi$ is
continuous. Then the problem (\ref{20A})
 with $k_0\in \widetilde{\mathcal{K}}_{\alpha_0}$
has a unique classical solution $k_t \in
\widetilde{\mathcal{K}}_{\alpha_t}$ with $t\in [0, T(\alpha_0,
\alpha))$, where $\alpha_t$ is the same as in Theorem~\ref{otm}.
\end{remark}
Now we show that the evolution of $k_t$ obtained in
Theorems~\ref{1tm} and~\ref{otm} can be continued in time to the
whole $\mathbb{R}_{+}$. Recall that the space $\mathcal{K}_\alpha$
was defined in (\ref{2D}) and $\varkappa> 0$ is the birth activity
parameter, see (\ref{R20}). Set
\begin{equation}
 \label{9D}
\alpha_\varkappa = - \ln \varkappa, \qquad K_\varkappa := \{ k \in
\mathcal{K}_{\alpha_\varkappa} :
 \|k\|_{\alpha_\varkappa} \leq 1\}.
\end{equation}
\begin{theorem}
 \label{2tm}
The solution of the problem (\ref{20A}) with $k_0 \in K_\varkappa$
can be continued to any positive $t$. Moreover, for every $t\geq 0$, the solution $k_t$
is also in ${K}_{_\varkappa}$.
\end{theorem}
Note that a correlation function $k$ is in $K_\varkappa$ if and only if $k(\eta) \leq \varkappa^{|\eta|}$ for $\lambda$-a.a. $\eta$.
Thus, we state that $k_0(\eta) \leq \varkappa^{|\eta|}$ implies
$k_t(\eta) \leq \varkappa^{|\eta|}$ for all $t\geq 0$.

The proof of Theorem~\ref{2tm} is based on the following estimate.
\begin{lemma}
 \label{alm}
Suppose that $k_0\in K_\varkappa$. For $\alpha < \alpha_\varkappa$,
let $k_t$ be the solution described by Theorem~\ref{otm}. Then, for
all $t \in [0, T(\alpha_\varkappa ,\alpha))$, we have that $k_t$ is
also in $K_\varkappa$.
\end{lemma}
\begin{proof}
For $G\in B_{\rm bs}(\Gamma_0)$ and $k_t$ as in Theorem~\ref{otm},
we set
\begin{equation}
 \label{12D}
 \varphi (t) = \langle \! \langle G, k_t \rangle \! \rangle = \int_{\Gamma} (KG) (\gamma) \mu_t (d \gamma),
\end{equation}
where $\mu_t$ is the corresponding state. According to (\ref{R20}), we have that
\begin{align}
 \label{13D}
 \frac{d}{d t} \varphi (t)  = & \int_{\Gamma} (LKG) (\gamma) \mu_t (d \gamma) \nonumber \\ = &
\int_{\Gamma} \sum_{x\in \gamma}
 \left[(K G) (\gamma \setminus x) - (K G) (\gamma) \right] \mu_t (d \gamma) \\
& + \varkappa \int_{\Gamma}\int_{\mathbb{R}^d}
 \left[(KG) (\gamma \cup x) - (K G) (\gamma) \right] \exp\left( - E^\phi (x , \gamma)\right)\mu_t (d \gamma).\nonumber
\end{align}
By (\ref{7A}),
\begin{align*}
(K G) (\gamma \setminus x) - (K G) (\gamma)&= \sum_{\xi \Subset
\gamma \setminus x} G (\xi) - \sum_{\xi \Subset \gamma \setminus x}
G (\xi) - \sum_{\xi \Subset \gamma \setminus x} G (\xi\cup x) \\&= -
\sum_{\xi \Subset \gamma \setminus x} G (\xi\cup x),
\end{align*}
and
\[
  (KG) (\gamma \cup x) - (K G) (\gamma)= \sum_{\xi \Subset \gamma}
G(\xi\cup x) + \sum_{\xi \Subset \gamma} G(\xi) - \sum_{\xi \Subset
\gamma} G(\xi)= \sum_{\xi \Subset \gamma} G(\xi \cup x ).
\]
Here all sums are finite since $G\in B_{\rm bs}(\Gamma_0)$.

Applying these equalities in (\ref{13D}) together with the following
\[
\sum_{x\in \gamma} \sum_{\xi \Subset \gamma\setminus x} G (\xi \cup
x) = \sum_{\xi \Subset \gamma} \sum_{x \in \xi} G(\xi) = \sum_{\xi
\Subset \gamma} |\xi| G (\xi),
\]
we arrive at
\begin{align}
 \label{14D}
\frac{d}{d t} \varphi (t) =& - \int_{\Gamma_0}|\eta| G(\eta) k_t
(\eta) \lambda (d \eta) \\&+ \varkappa \int_{\Gamma}
\left(\int_{\mathbb{R}^d} \sum_{\xi \Subset \gamma} G (\xi\cup x) d
x \right) \exp\left( - E^\phi (x , \gamma)\right)\mu_t (d \gamma).
\nonumber
\end{align}
Assume now that $G$ is positive as an element of $L^1 (\Gamma_0,
d\lambda)$. As $\phi$ in (\ref{14A}) is also positive, the second
line in (\ref{14D}) can be estimated
\begin{align*}
 &  \,\varkappa \int_{\Gamma}
 \left(\int_{\mathbb{R}^d} \sum_{\xi \Subset \gamma} G (\xi\cup x) d x \right)
\exp\left( - E^\phi (x , \gamma)\right)\mu_t (d \gamma) \\ \leq & \,
\varkappa \int_{\mathbb{R}^d}\left( \int_{\Gamma}
 \sum_{\xi \Subset \gamma} G (\xi\cup x) \mu_t (d \gamma) \right) d x =\varkappa \int_{\mathbb{R}^d}\left( \int_{\Gamma_0}G (\eta\cup x) k_t (\eta)\lambda (d \eta)\right) d x
\\   =& \,\varkappa \int_{\Gamma_0}G (\eta) \sum_{x \in \eta} k_t (\eta \setminus x) \lambda (d \eta).
\end{align*}
In obtaining the latter equality we have used (\ref{12A}). Applying
this estimate in (\ref{14D}) we arrive at
\begin{equation}
 \label{15D}
\frac{d}{dt} k_t (\eta) \leq - |\eta| k_t (\eta) + \varkappa \sum_{x
\in \eta} k_t (\eta \setminus x),
\end{equation}
which holds for $\lambda$-almost all $\eta \in \Gamma_0$. By standard methods we get from this
\begin{equation}
 \label{16D}
k_t (\eta) \leq k_0 (\eta) e^{-|\eta|t} + \varkappa \int_0^t e^{- (t-s) |\eta|}\sum_{x \in \eta} k_s (\eta \setminus x)ds.
\end{equation}
As a correlation function, $k_t (\eta) \geq 0$ for $\lambda$-almost
all $\eta$. For a fixed $t$ and $|\eta|= n$, assume that
$k_s(\xi)\leq \varkappa^{|\xi|}$ for all $s\in [0,t]$ and all $|\xi|
= n-1$. As we also assume $k_0(\eta)\leq \varkappa^{|\eta|}$, by
(\ref{16D}) $k_s(\eta)\leq \varkappa^{|\eta|}$, also for all $s\in
[0,t]$. Thus, $k_t \in K_\varkappa$.
\end{proof}
\begin{proof}[Proof of Theorem~\ref{2tm}.] Take any $\alpha<
\alpha_\varkappa$. Then the solution $k_t$ exists in
$\mathcal{K}_\alpha$ for $t\in [0, T (\alpha_\varkappa,\alpha))$. If
$k_0$ is in $K_\varkappa$, by Lemma~\ref{alm} $k_t$ is also in
$K_\varkappa$. We take any $\tau \in [0, T (\alpha_\varkappa,
\alpha))$ and consider the problem (\ref{20A}) for $\tilde{k}_t$
with the initial condition $\tilde{k}_0 = k_\tau \in K_\varkappa$.
This gives the continuation in question to $[0, \tau+ T
(\alpha_\varkappa, \alpha))$. Then we repeat the above arguments.
\end{proof}

\section{The mesoscopic description}

The mesoscopic description of the dynamics of our model will be
conducted in the Vlasov scaling framework, see~\cite{DimaN} where
the detailed presentation of this approach and the most updated
related bibliography can be found. Our program in this section is as
follows:
\begin{itemize}
 \item Derive the Vlasov hierarchy as the (Vlasov) scaling limit of (\ref{20A}) and prove the existence of its solutions, see (\ref{E12}) and Proposition~\ref{Epn}.
\item Then derive the Vlasov equation from the Vlasov hierarchy and prove the existence of its solutions, see (\ref{E14}) and Theorem~\ref{Vatm}.
\item Prove the convergence of the rescaled correlation functions to the solutions of the Vlasov hierarchy, see Theorem~\ref{3tm}.
\end{itemize}

\subsection{The Vlasov hierarchy}

In the Vlasov scaling limit, which is achieved by letting
$\varepsilon \rightarrow 0$, the particle density is supposed to
diverge whereas the interaction gets weak and of long range. Thus,
we assume that the correlation function $k^{(\varepsilon)}$ of the
particle system we consider depends of the scaling parameter and
diverges in such a way that the renormalized function
\begin{equation}
\label{E1}
 k^{(\varepsilon)}_{\rm ren} (\eta) = \varepsilon^{|\eta|} k^{(\varepsilon)}(\eta),
\end{equation}
has a finite limit, which we denote by $r$. The evolution of yet not rescaled
functions is described by the equation (\ref{20A}) in which the generator $L^\Delta_\varepsilon$
contains $\varepsilon \phi$ and $\varepsilon^{-1} \varkappa$ in place of $\phi$ and $\varkappa$, respectively.
Thus, the evolution
of the rescaled functions (\ref{E1}) is described by the equation
\begin{equation}
\label{E2}
 \frac{d}{dt} k^{(\varepsilon)}_{t,\rm ren} = L_{\varepsilon, {\rm ren}} k^{(\varepsilon)}_{t,\rm ren},
\qquad k^{(\varepsilon)}_{t,\rm ren}|_{t=0} = k^{(\varepsilon)}_{0,\rm ren},
\end{equation}
where
\begin{equation}
 \label{E3}
L_{\varepsilon, {\rm ren}} = R_{\varepsilon}L_{\varepsilon}^\Delta
R_{\varepsilon}^{-1},
 \qquad \left( R_{\varepsilon}k\right) (\eta) := \varepsilon^{|\eta|} k(\eta).
\end{equation}
By (\ref{21A}), we thus have
\begin{align}
 \label{E4}
\left( L_{\varepsilon, {\rm ren}} k\right) (\eta) = & - |\eta|
k(\eta) \\  & +  \varkappa \sum_{x\in \eta}
e(\tau^{(\varepsilon)}_x, \eta \setminus x) \int_{\Gamma_0}
e(\varepsilon^{-1}t^{(\varepsilon)}_x, \xi) k(\eta \setminus x \cup
\xi) \lambda (d\xi), \nonumber
\end{align}
with
\begin{equation}
 \label{E5}
\tau_x^{(\varepsilon)}(y) := \exp\left[ - \varepsilon \phi(x-y)\right], \qquad t^{(\varepsilon)}_x := \tau^{(\varepsilon)}_x
-1.
\end{equation}
Note that, for small enough $\varepsilon$, $R_\varepsilon^{-1} k$
might not be a correlation function, even if $k$ is, see (\ref{E3}).

For any $\alpha_0$, $\alpha'$, $\alpha''$, and $\alpha$ such that $\alpha < \alpha' < \alpha''< \alpha_0$, as in (\ref{o7}) we get
\begin{equation}
 \label{E6}
\|L_{\varepsilon, {\rm ren}} \|_{\alpha'' \alpha'} \leq
\frac{1}{(\alpha'' - \alpha')e} \left[1 + \varkappa \exp\left(
\alpha_0 + c_\phi^{(\varepsilon)} e^{-\alpha} \right) \right],
\end{equation}
where, c.f. (\ref{13A}),
\begin{equation}
 \label{E7}
 c_\phi^{(\varepsilon)} = \varepsilon^{-1} \int_{\mathbb{R}^d} \left(1 - e^{-\varepsilon \phi(x)} \right)dx.
\end{equation}
Suppose now that $\phi$ is in $L^1(\mathbb{R}^d)$ and set
\begin{equation}
 \label{E8}
 \langle \phi \rangle = \int_{\mathbb{R}^d} \phi(x) dx.
\end{equation}
Recall that we still assume $\phi \geq 0$. Then
\begin{align}
 \label{E9}
 \|L_{\varepsilon, {\rm ren}} \|_{\alpha'' \alpha'} & \leq \sup_{\varepsilon >0} \{{\rm RHS(\ref{E6})}\} \\  & =
 \frac{1}{(\alpha'' - \alpha')e}
\left[1 + \varkappa \exp\left( \alpha_0 + \langle \phi \rangle e^{-\alpha} \right) \right]. \nonumber
\end{align}
Now let us informally pass in (\ref{E4}) to the limit $\varepsilon
\rightarrow 0$. We then obtain the following operator
\begin{align}
 \label{E10} \quad
\left( L_V k\right) (\eta)  = & - |\eta| k(\eta) \\  & +  \varkappa
\sum_{x\in \eta} \int_{\Gamma_0} e(- \phi(x-\cdot), \xi) k(\eta
\setminus x \cup \xi) \lambda (d\xi). \nonumber
\end{align}
It certainly obeys
\begin{equation}
 \label{E11}
 \|L_V \|_{\alpha'' \alpha'} \leq \frac{1}{(\alpha'' - \alpha')e}
\left[1 + \varkappa \exp\left( \alpha_0 + \langle \phi \rangle e^{-\alpha} \right) \right],
\end{equation}
 and hence along with the problem (\ref{E2}) we can consider
 \begin{equation}
 \label{E12}
 \frac{d}{dt} r_t = L_V r_t, \qquad r_t|_{t=0} = r_0,
 \end{equation}
which is called the {\it Vlasov hierarchy} for the Glauber dynamic we consider.
Set, c.f. (\ref{7}),
\begin{equation}
 \label{7QQ}
\widetilde{T}(\alpha_0, \alpha) := \frac{\alpha_0 - \alpha}
{1 + \varkappa \exp\left(\alpha_0 + \langle \phi \rangle e^{-\alpha} \right)} \leq {T}(\alpha_0, \alpha).
\end{equation}
The latter inequality holds as $c_\phi^{(\varepsilon)} \leq \langle
\phi \rangle$, see (\ref{E7}) and (\ref{E8}). Repeating the
arguments used in the proof of Theorem~\ref{otm} we obtain the
following
\begin{proposition}
 \label{Epn}
Let $\phi$, $\alpha_0$, $\alpha$, be as in Theorem~\ref{otm} and
$\widetilde{T}(\alpha_0,\alpha)$ be as in (\ref{7QQ}). Then the
problem (\ref{E2}) (resp. (\ref{E12})) with any $\varepsilon >0$ and
$k_{0, {\rm ren}}^{(\varepsilon)}\in \mathcal{K}_{\alpha_0}$ (resp.
$r_0 \in \mathcal{K}_{\alpha_0}$) has a unique classical solution
$k_{t, {\rm ren}}^{(\varepsilon)} \in \mathcal{K}_{\alpha}$ (resp.
$r_t \in \mathcal{K}_{\alpha}$) with $t\in [0,
\widetilde{T}(\alpha_0, \alpha))$.
\end{proposition}
Note that the passage from (\ref{E4}) to (\ref{E10}) was only
`informal', so we have no information how `close' is $r_t$ to $k_{t,
{\rm ren}}^{(\varepsilon)}$. Another observation is that (\ref{E12})
has a very special solution, which we obtain now.

\subsection{The Vlasov equation}

For the potential $\phi$ and an appropriate function $g$, we write
\begin{equation}
 \label{E13}
 (\phi \ast g) (x) = \int_{\mathbb{R}^d} \phi(x-y) g(y) dy.
\end{equation}
Let us consider in $L^\infty (\mathbb{R}^d)$ the following problem,
c.f.~\cite[Example 8]{DimaN},
\begin{equation}
 \label{E14}
\frac{d}{dt}\varrho_t(x) = - \varrho_t (x) + \varkappa \exp\left( - (\phi \ast \varrho_t) (x)\right), \quad \varrho_t|_{t=0} = \varrho_0.
\end{equation}
Given $\alpha \in \mathbb{R}$, we denote
\begin{align}
 \label{Va1}
 \varDelta_\alpha & =  \{ \varrho \in L^\infty (\mathbb{R}^d) : \|\varrho \|_{L^\infty(\mathbb{R}^d)} \leq e^{-\alpha}\}, \\
 \varDelta^{+}_\alpha & =  \{ \varrho \in \varDelta_\alpha : \varrho (x) \geq 0 \quad {\rm a. e}\}. \nonumber
\end{align}
\begin{lemma}
 \label{Vlm}
Suppose that, for some $\alpha_0\in \mathbb{R}$ and $T>0$, the
problem (\ref{E14}) with $\varrho_0 \in \varDelta^{+}_{\alpha_0}$
has a unique classical solution $\varrho_t\in
\varDelta^{+}_{\alpha_0}$ on the time interval $[0,T)$. Then the
solution $r_t\in \mathcal{K}_{\alpha}$, $\alpha < \alpha_0$, of the
problem (\ref{E12}), as in Proposition~\ref{Epn}, with $r_0 (\eta )
= e(\varrho_0 , \eta)\in \mathcal{K}_{\alpha_0}$ has the form
\begin{equation}
 \label{C23}
r_t (\eta) = e(\varrho_t, \eta)= \prod_{x\in \eta} \varrho_t(x),
\end{equation}
and hence remains in $\mathcal{K}_{\alpha_0}$.
\end{lemma}
\begin{proof}
First of all we note that $e(\varrho, \cdot) \in
\mathcal{K}_{\alpha_0}$ if and only if $\varrho \in
\varDelta_{\alpha_0}$, see (\ref{o5}). Now set $\tilde{r}_t =
e(\varrho_t, \cdot)$ with $\varrho_t$ solving (\ref{E14}). This
$\tilde{r}_t$ solves (\ref{E12}), which can easily be checked by
computing $d/dt$ and employing (\ref{E14}). In view of the
uniqueness as in Proposition~\ref{Epn}, we then have $\tilde{r}_t =
r_t$ on the time interval where both solutions exist.\end{proof}
\begin{remark}
 \label{Vrk}
As
(\ref{C23}) is the correlation function for the Poisson measure $\pi_{\varrho_t}$, see (\ref{3Aa}) and (\ref{3Ab}), the property established by
the above lemma can be called the {\it chaos preservation}. Indeed, the most chaotic state of the system is the free state described
 by a Poisson measure.
\end{remark}
Let us show now that the problem (\ref{E14}) does have the solution
we need (c.f. \cite[Theorem 3.3]{FKK}). In a standard way, this
problem can be transformed into the following integral equation
\begin{equation}
 \label{E15}
\varrho_t (x) = \varrho_0(x) e^{-t} + \varkappa \int_0^t e^{-(t-s)} \exp \left( - (\phi\ast \varrho_s )(x)\right)ds.
\end{equation}
Following classical Picard's scheme we seek the solution as the limit of the iterative sequence $\{\varrho^{(n)}_t\}_{n\in \mathbb{N}_0}$,
defined as
\begin{equation}
 \label{E16}
 \varrho^{(n)}_t (x) = \varrho_0(x) e^{-t} + \varkappa \int_0^t e^{-(t-s)} \exp \left( - (\phi\ast \varrho^{(n-1)}_s )(x)\right)ds, \quad n\in \mathbb{N},
\end{equation}
and $\varrho^{(0)}_t = \varrho_0$. Clearly $\varrho^{(n)}_t \geq 0$ for all $n\in \mathbb{N}_0$. Thus, we have to
show that $\varrho^{(n)}_t (x) \leq e^{-\alpha_0}$, at least for some $t>0$. By the induction over $n$, we see that this holds,
for all $t>0$,
if
\begin{equation}
 \label{Ea16}
\varkappa \leq e^{-\alpha_0}.
\end{equation}
Now let us show that $\{\varrho^{(n)}_t\}_{n\in \mathbb{N}_0}$ is a
Cauchy sequence in $L^\infty(\mathbb{R}^d)$,
 assuming $\varrho^{(n)}_s\in \varDelta_{\alpha_0}$, for all $n\in \mathbb{N}_0$ and $s\leq t$.
From (\ref{E16}), using elementary inequality
\mbox{$|e^{-a}-e^{-b}|\leq|a-b|$} for $a,b\geq0$, we get
\begin{align*}
\| \varrho^{(n)}_t - \varrho^{(n-1)}_t\|_{L^\infty (\mathbb{R}^d)}
&\leq  q(t) \sup_{s\in [0,t]}
 \| \varrho^{(n-1)}_s - \varrho^{(n-2)}_s\|_{L^\infty
 (\mathbb{R}^d)},\\\intertext{where}
 q(t) & :=  \varkappa \langle \phi \rangle \left(1 - e^{-t} \right)
 . 
\end{align*}
Now we take $T>0$ such that $q(T) < 1$. Then the latter estimate
yields
\begin{equation}
 \label{E17}
 \sup_{t\in [0,T]} \| \varrho^{(n)}_t - \varrho^{(n-1)}_t\|_{L^\infty (\mathbb{R}^d)} \leq q(T)
\sup_{t\in [0,T]} \| \varrho^{(n-1)}_t -
\varrho^{(n-2)}_t\|_{L^\infty (\mathbb{R}^d)}
\end{equation}
Therefore, the sequence $\{\varrho^{(n)}_t\}_{n\in \mathbb{N}_0}$
converges in $L^\infty(\mathbb{R}^d)$, uniformly on $[0,T]$. Thus,
its limit is the unique classical solution of (\ref{E14}). Since
this limit is still in $\varDelta_{\alpha_0}^+$, the evolution can
be continued. Taking into account Lemma~\ref{Vlm} we come to the
following conclusion.
\begin{theorem}
 \label{Vatm}
Given $\varkappa>0$, let $\alpha_0$ be as in (\ref{Ea16}). Then the unique classical solution of (\ref{E12}) with $r_0 = e(\varrho_0, \cdot)$, $\varrho_0 \in \varDelta_{\alpha_0}^{+}$,
exists for all $t>0$ and is given by (\ref{C23}) with $\varrho_t \in \varDelta_{\alpha_0}^{+}$ being the solution of (\ref{E14}).
\end{theorem}

\subsection{The scaling limit $\varepsilon \rightarrow 0$}

Our final task in this work is to show that the solution of
(\ref{E2}) $k_t^{(\varepsilon)}$ converges in
$\mathcal{K}_{\alpha_0}$ uniformly in on $[0,T]$, $T < T(\alpha_0,
\alpha)$, to the solution of (\ref{E12}), see Proposition~\ref{Epn}.
Here we should impose an additional condition on the potential
$\phi$, which, however, seems quite natural. Recall that in this
section we suppose $\phi\in L^1 (\mathbb{R}^d)$.
\begin{theorem}
 \label{3tm}
Let $\phi$, $\alpha_0$, $\alpha$, and
$\widetilde{T}(\alpha_0,\alpha)$ be as in Proposition~\ref{Epn}.
Assume also that $\phi \in L^1 (\mathbb{R}^d) \cap L^\infty
(\mathbb{R}^d)$ and consider the problems (\ref{E2}) and (\ref{E12})
with $k^{(\varepsilon)}_{0, {\rm ren}} = r_0 \in
{\mathcal{K}}_{\alpha_0}$. For their solutions
$k^{(\varepsilon)}_{t, {\rm ren}}$ and $r_t$, it follows that
$k^{(\varepsilon)}_{t, {\rm ren}} \rightarrow r_t$ in
${\mathcal{K}}_{\alpha}$, as $\varepsilon \rightarrow 0$, uniformly
on every $[0,T]$, $T < \widetilde{T}(\alpha_0, \alpha)$.
\end{theorem}
\begin{proof}
Given $n\in \mathbb{N}$, let $k^{(\varepsilon)}_{t,n}$ and $r_{t,n}$
be defined as in (\ref{3D}) with $L_{\varepsilon, {\rm ren}}$ and
$L_V$, respectively. Like in the proof of Theorem~\ref{otm}, one can
show that the sequences of $k^{(\varepsilon)}_{t,n}$ and $r_{t,n}$
converge in $\mathcal{K}_{\alpha}$ to $k^{(\varepsilon)}_{t}$ and
$r_{t}$, respectively, uniformly on every $[0,T]$, $T < T(\alpha_0,
\alpha)$. Then, for $\delta >0$, one finds $n\in \mathbb{N}$ such
that, for all $t\in[0,T]$,
\begin{equation}
 \label{U}
\|k^{(\varepsilon)}_{t,n} - k^{(\varepsilon)}_{t,{\rm ren}}\|_{\alpha} + \| r_{t,n} - r_t\|_{\alpha} < \delta/2.
\end{equation}
Then in view if (\ref{3D}),
\begin{align}
 \label{U1} \qquad
 \|k^{(\varepsilon)}_{t,{\rm ren}} - r_t\|_{\alpha} & \leq  \bigg{\|}\sum_{m=1}^n
\frac{1}{m!} t^m \left(L_{\varepsilon, {\rm ren}}^m - L^m_V \right)r_0\bigg{\|}_{\alpha}
 + \frac{\delta}{2} \\  & \leq
\|L_{\varepsilon, {\rm ren}} - L_V\|_{\alpha_0 \alpha} \|r_0\|_{\alpha_0} T \exp \left(T b(\alpha_0, \alpha)\right)
 + \frac{\delta}{2}, \nonumber
\end{align}
where, see (\ref{o7}),
\[
b(\alpha_0 , \alpha) := \frac{1}{(\alpha_0 -\alpha)e}
 \left[1 + \varkappa \exp\left(\alpha_0+ \langle \phi \rangle e^{-\alpha} \right) \right].
\]
Here we used the following representation
\begin{align}
 \label{U2}
  L_{\varepsilon, {\rm ren}}^m - L^m_V = & \left(L_{\varepsilon,
{\rm ren}} - L_V \right) L_{\varepsilon, {\rm ren}}^{m-1} + L_V
\left(L_{\varepsilon, {\rm ren}} - L_V \right) L_{\varepsilon, {\rm
ren}}^{m-2} \\
 &  + \cdots + L_V^{m-2} \left(L_{\varepsilon, {\rm ren}} - L_V \right)
 L_{\varepsilon, {\rm ren}} + L_V^{m-1} \left(L_{\varepsilon, {\rm ren}} - L_V \right) .\nonumber
\end{align}
Thus, we have to show that
\begin{equation}\label{con}
 \|L_{\varepsilon, {\rm ren}}
-L_V\|_{\alpha_0\alpha}\to 0, \quad \ {\rm as} \ \ \varepsilon\to 0,
\end{equation}
which will allow us to make the first summand in the right-hand side of (\ref{U1}) also smaller than $\delta/2$ and thereby to complete the proof.

Subtracting (\ref{E10}) from (\ref{E4}) we get
\begin{equation}
 \label{Vx}
 \left(L_{\varepsilon, {\rm ren}}
-L_V \right)k(\eta) = \varkappa \sum_{x\in \eta} \int_{\Gamma_0}
Q_\varepsilon (x, \eta\setminus x, \xi) k(\eta \setminus x \cup \xi)
\lambda (d\xi) \end{equation} where
\begin{align}
 \label{Vx1}
Q_\varepsilon (x, \eta\setminus x, \xi)  :&=
e(\tau_x^{(\varepsilon)}, \eta \setminus x)
e(\varepsilon^{-1}t^{(\varepsilon)}_x,\xi)
 - e (- \phi(x - \cdot), \xi)\\
 &= e(\varepsilon^{-1}t^{(\varepsilon)}_x,\xi) - e (- \phi(x - \cdot)-\left[ 1 - e(\tau_x^{(\varepsilon)}, \eta \setminus x) \right] e(\varepsilon^{-1}t^{(\varepsilon)}_x,\xi). \nonumber
\end{align}
For $t>0$, the function $e^{-t} - 1 + t$ takes positive values only;
hence
\[
\Psi(t) : = (e^{-t} - 1 + t)/t^2, \quad t>0,
\]
is positive and bounded, say by $C>0$.
Then by means of the following elementary analog of (\ref{U2})
\[
b_1 \cdots b_n - a_1 \cdots a_n \leq \sum_{i=1}^n (b_i - a_i) b_1
\cdots b_{i-1} b_{i+1} \cdots b_n, \quad b_i \geq a_i >0,
\]
we obtain
\begin{align*}
 \left\vert e(\varepsilon^{-1}t^{(\varepsilon)}_x,\xi) - e (- \phi(x - \cdot) \right\vert &\leq \sum_{y\in \xi} \varepsilon [\phi(x-y)]^2
\Psi\left(\varepsilon \phi(x-y)\right)  \prod_{z\in \xi\setminus y}
\phi(x-z) \\   &\leq\varepsilon C \sum_{y\in \xi}
 [\phi(x-y)]^2 e(\phi(x-\cdot), \xi\setminus y),
\end{align*}
and
\begin{align*}
 \left\vert \left[ 1 - e(\tau_x^{(\varepsilon)}, \eta \setminus x) \right] e(\varepsilon^{-1}t^{(\varepsilon)}_x,\xi) \right\vert \leq \varepsilon
 \sum_{y\in \eta\setminus x} \phi(x-y) e(\phi(x-\cdot), \xi).
\end{align*}
Then from (\ref{Vx}) for $\lambda$-almost all $\eta$ we have, see (\ref{o5}),
\begin{align*}
 &\left\vert ( L_{\varepsilon, {\rm ren}}
-L_V) k(\eta)\right\vert \\ \leq &\, \varkappa \|k\|_{\alpha_0}
e^{-\alpha_0 |\eta|} \sum_{x\in \eta} \int_{\Gamma_0} \exp( -
\alpha_0 |\xi| + \alpha_0)\\  & \quad\times \bigg{\{} \varepsilon C
\sum_{y \in \xi} [\phi(x-y)]^2 e(\phi(x-\cdot), \xi\setminus y) +
\varepsilon \sum_{y\in \eta\setminus x}\phi(x-y)
 e(\phi(x-\cdot), \xi) \bigg{\}} \lambda (d \xi) \displaybreak[2]\\   \leq &\,
\varkappa \varepsilon \|k\|_{\alpha_0} e^{-\alpha_0 |\eta|}
\sum_{x\in \eta} \int_{\Gamma_0} e^{-\alpha_0 |\xi|}e(\phi(x-\cdot),
\xi) \\
& \quad \times \bigg{\{} C \int_{\mathbb{R}^d} [\phi(x-y)]^2 dy +
e^{\alpha_0}
 \sum_{y\in \eta\setminus x} \phi(x-y) \bigg{\}}\lambda (d \xi)
\\ \leq &\, \varkappa \varepsilon \|k\|_{\alpha_0}
\|\phi\|_{L^\infty (\mathbb{R}^d)}
 \exp\left( \langle \phi \rangle e^{-\alpha_0}\right)\left[C \langle \phi \rangle
 |\eta| + e^{\alpha_0} |\eta| ( |\eta|-1) \right]e^{-\alpha_0 |\eta|}.
\end{align*}
This yields
\begin{align}
 \label{Vx2}
 \|L_{\varepsilon, {\rm ren}}
-L_V\|_{\alpha_0\alpha} \leq &\, \varepsilon \varkappa
\|\phi\|_{L^\infty(\mathbb{R}^d)}
 \exp\left( \langle \phi \rangle e^{-\alpha_0}\right) \\
&\times  \left[\frac{C \langle \phi \rangle }{( \alpha_0 - \alpha)e}
+ \frac{4 e^{\alpha_0}} {( ( \alpha_0 - \alpha)e)^2} \right],
\nonumber
\end{align}
and thereby (\ref{con}).
\end{proof}

\section{Concluding remarks}

Regarding the evolution of quasi-observables, in Theorem~\ref{1tm}
we have proven its existence in $\mathcal{G}_\alpha$ if $G_0 \in
\mathcal{G}_{\alpha_0}$, for any $\alpha_0$ and any
$\alpha>\alpha_0$, and
 for all values
of the model parameters $c_\phi$ and $\varkappa$, however, on a
bounded time interval. Note that the bound $T(\alpha, \alpha_0)$ is
small for big $c_\phi \varkappa$, see (\ref{7R}). Note also that
there exists the scale of spaces $\mathcal{G}_{\alpha_t} \subset
\mathcal{G}_\alpha$ such that $G_t \in \mathcal{G}_{\alpha_t}$ for
$t\in [0,T(\alpha, \alpha_0))$, similarly to Theorem~\ref{otm}. For
$c_\phi \varkappa<1/e$, the evolution $G_0 \mapsto G_t$ is described
by a $C_0$-semigroup, and hence has no time bounds, see
Theorem~\ref{20tm}.

Turn now to the evolution of states and correlation functions. The
main peculiarity of Theorem~\ref{otm} is that, in contrast to the
results of~\cite{Oles1,Oles,OlesLena}, here we (a) impose no
restrictions on $c_\phi$ and $\varkappa$; (b) describe the evolution
directly, not as a weak evolution via (\ref{1D}). The price is the
time restriction, similar as in Theorem~\ref{1tm}. Again, we can
start in $\mathcal{K}_{\alpha_0}$ with any $\alpha_0\in \mathbb{R}$,
and obtain that $k_t \in \mathcal{K}_{\alpha_t} \subset
\mathcal{K}_\alpha$, also for any $\alpha < \alpha_0$. The time
bound $T(\alpha_0, \alpha)$ depends on the choice of $\alpha_0$ and
$\alpha$. If the initial states is dominated by the Poisson measure
with intensity $\varkappa$, that is, if $k_0 (\eta) \leq
\varkappa^{|\eta|}$, then the solution described by
Theorem~\ref{otm} has also the property $k_t (\eta) \leq
\varkappa^{|\eta|}$, and hence can be continued in time ad
infinitum, see Theorem~\ref{2tm}. Of course, in this case $\alpha_0$
should obey (\ref{Ea16}). The main aim of using the Vlasov hierarchy
(\ref{E12}) is obtaining the scaling limit of the rescaled
correlation functions $k^{(\varepsilon)}_{t, {\rm ren}}$. For any
$\alpha_0\in \mathbb{R}$ and $r_0 \in \mathcal{K}_{\alpha_0}$, this
hierarchy has a unique classical solution $r_t$ in any
$\mathcal{K}_\alpha$, $\alpha < \alpha_0$, with $t\in [0,
\widetilde{T}(\alpha_0, \alpha))$, see Proposition~\ref{Epn}. Here,
however, for general $r_0$ we have no tools for continuing $r_t$,
like we did in Theorem~\ref{2tm} where we used the connection of
$L^\Delta$ with $L$ given by (\ref{R20}), since neither Markov
operator corresponds to $L_V$. But if $r_0$ is Poissonian, i.e.,
$r_0= e(\varrho_0, \cdot)$, then (\ref{E12}) has the solution $r_t=
e(\varrho_t, \cdot)$ with infinite time lives in `sufficiently
large' $\mathcal{K}_{\alpha_0}$, see Theorem~\ref{Vatm}. The latter
means that the Poissonian correlation function $k(\eta) =
\varkappa^{|\eta|}$ belongs to this $\mathcal{K}_{\alpha_0}$, see
(\ref{Ea16}).

Note also that in the recent paper \cite{FKK} it was shown the
existence and strong convergence in the Vlasov scaling for the
classical solution in one space $\mathcal{K}_\alpha$ but again under
the condition $c_\phi \varkappa<1/e$.

\vskip.5cm

\noindent {\it Acknowledgment:} The authors are grateful to Oles
Kutovyi for valuable discussions.

\bibliographystyle{amsplain}

\end{document}